\documentclass{article}

\usepackage{mmp}
\usepackage{bbm}
\usepackage{thm-restate}

\definecolor{DarkGreen}{rgb}{0.1,0.5,0.1}

\newcommand{\todo}[1]{\textcolor{DarkGreen}{[To do: #1]}}
\newcommand{\rc}[1]{\textcolor{red}{[Rachel: #1]}}
\newcommand{\mmp}[1]{\textcolor{red}{[MMP: #1]}}

\graphicspath{{plots/}}

\newcommand{\bhat}{\hat{b}}
\newcommand{\va}{v_A}

\newcommand{\upi}{\underline{\pi}}
\newcommand{\bpi}{\bar{\pi}}

\title{\textsc{The Strange Case of Privacy in Equilibrium Models}}
\author{%
\textsc{Rachel Cummings}\thanks{Computing and Mathematical Sciences, California Institute of Technology, Email: \href{mailto:rachelc@caltech.edu}{rachelc@caltech.edu}. Supported in part by a Simons Award for Graduate Students in Theoretical Computer Science, NSF grant CNS-1254169, US-Israel Binational Science Foundation grant 2012348, and a Google Faculty Research Award.}
\and
\textsc{Katrina Ligett}\thanks{Computing and Mathematical Sciences, California Institute of Technology, Email: \href{katrina@caltech.edu}{katrina@caltech.edu}. Supported in part by NSF grant CNS-1254169, US-Israel Binational Science Foundation grant 2012348, the Charles Lee Powell Foundation, a Google Faculty Research Award, an Okawa Foundation Research Grant, and a Microsoft Faculty Fellowship.}
\and
\textsc{Mallesh M. Pai}\thanks{Department of Economics, University of Pennsylvania. Email: \href{mailto:mallesh@econ.upenn.edu}{mallesh@econ.upenn.edu}. Supported in part by NSF Grant CCF-1101389.}
\and
\textsc{Aaron Roth}\thanks{Department of Computer and Information Science, University of Pennsylvania. Email: \href{mailto:aaroth@cis.upenn.edu}{aaroth@cis.upenn.edu}. Supported in part by NSF Grant CCF-1101389, an NSF CAREER award, and an Alfred P. Sloan Foundation Fellowship.}
}
\begin{document}
\maketitle
\begin{abstract}
We study how privacy technologies affect user and advertiser behavior in a simple economic model of targeted advertising. In our model, a consumer first decides whether or not to buy a good, and then an advertiser chooses an advertisement to show the consumer. The consumer's value for the good is correlated with her type, which determines which ad the advertiser would prefer to show to her---and hence, the advertiser would like to use information about the consumer's purchase decision to target the ad that he shows.

In our model, the advertiser is given only a differentially private signal about the consumer's behavior---which can range from no signal at all to a perfect signal, as we vary the differential privacy parameter. This allows us to study equilibrium behavior \emph{as a function of the level of privacy provided to the consumer}. We show that this behavior can be highly counter-intuitive, and that the effect of adding privacy in \emph{equilibrium} can be completely different from what we would expect if we ignored equilibrium incentives. Specifically, we show that increasing the level of privacy can actually \emph{increase} the amount of information about the consumer's type contained in the signal the advertiser receives, lead to decreased utility for the consumer, and increased profit for the advertiser, and that generally these quantities can be non-monotonic and even discontinuous in the privacy level of the signal.
\end{abstract}

\newpage
\clearpage
\setcounter{page}{1}

\section{Introduction}
As advertising becomes increasingly targeted and data driven, there is growing concern that the algorithms driving these personalization decisions can be \emph{discriminatory}. For example, \cite{DTD15} highlighted potential gender bias in Google's advertising targeting algorithms, giving evidence that male job seekers were more likely to be shown ads for high paying jobs than female job seekers.  Similarly, the FTC has expressed concern that data mining of user online behavior, which data brokers use to categorize users into categories such as ``ethnic second-city struggler'' and ``urban scrambler'', is used to selectively target users for high interest loans \citep{newsarticle}.

One tempting response to such concerns is regulation: for example, we could mandate the use of privacy technologies which would explicitly limit the amount of information advertisers could learn about users past behavior in a quantifiable way.\footnote{An alternative approach to limiting the information that advertisers can collect is to explicitly try to limit how they use that information to avoid unfairness. See \cite{DHPRZ12} for work in this direction.} If advertisers are only able to see differentially private signals about user behavior, for example, then we can precisely quantify the amount of information that the advertiser's signal contains about the actions of the user. As we increase the level of privacy, we would naively expect to see several effects: first, the amount of information that the advertiser learns about the user should decrease. Second, the utility of the advertiser should decrease, since she is now less able to precisely target her advertisements. Finally, if the user really was experiencing disutility from the way that the advertiser had been targeting her ads, the user's utility should increase.

These expectations are not necessarily well grounded, however, for the following reason: in strategic settings, as the information content of the signal that the advertiser receives changes, he will change the way he uses the information he receives to target ads. Similarly, given the way that user behavior is used in ad targeting, a sophisticated user may change her browsing behavior. Hence, it is not enough to statically consider the effect of adding privacy technologies to a system, but instead we must consider the effect in equilibrium. Each privacy level defines a different strategic interaction which results in different equilibrium behavior, and a-priori, it is non-obvious the effect that the privacy technology will have on the equilibrium outcome.

In this paper, we consider a very simple two-stage model of advertising that may be targeted based on past user behavior.  The model has three rational agents: a consumer, a seller, and an advertiser. The consumer has a value for the good being sold by the seller, and also a type which determines which of several ads the advertiser benefits most from showing the consumer. The consumer's value and type are drawn from a joint distribution, and the two are correlated --- hence, information about the consumer's value for the good being sold is relevant to the advertiser when determining what ad to show her.

The game proceeds in two stages. In the first stage, the seller determines a price to set for the good he is selling --- then the consumer determines whether or not she wishes to buy the good. In the second stage, the advertiser receives some signal about the purchase decision of the consumer in the first round. The signal may be noisy; the correlation of the signal with the consumer's purchase decision reflects the level of privacy imposed on the environment, and is quantified via differential privacy. As a function of the signal, the advertiser performs a Bayesian update to compute his posterior belief about the consumer's type, and then decides which ad to show the consumer. The consumer has a preference over which ad she is shown --- for example, one might be for a credit card with a lower interest rate, or for a job with a higher salary. Hence, the consumer does not necessarily act myopically in the first round when deciding whether to purchase the seller's good or not, and instead takes into account the effect that her purchase decision will have on the second round.

We characterize the equilibria of this model as a function of the level of differential privacy provided by the signal the advertiser receives. We show that in this model, several counter-intuitive phenomena can arise.  For example, the following things may occur in equilibrium as we increase the privacy level (i.e. decrease the correlation between the user's purchase decision and the signal received by the advertiser):
\begin{enumerate}
\item The signal received by the advertiser can actually contain \emph{more} information about the agent's type, as measured by mutual information (Figure \ref{fig.mutualinfo}). Similarly, the difference in the advertiser's posterior belief about the true type of the agent after seeing the purchase/ did not purchase noisy bits can increase (Figure \ref{fig.equilboundaries}). (Interestingly, this difference does not necessarily peak at the same privacy level as the mutual information between the consumer's type and the advertiser's signal).
\item Consumer utility can decrease, and advertiser utility can increase (Figure \ref{fig.advertiserutility}, \ref{fig.csandprofit}).
\item More generally, these quantities can behave in complicated ways as a function of the privacy parameter $\epsilon$: they are not necessarily monotone in $\epsilon$, and can even be discontinuous, and equilibrium multiplicity can vary with $\epsilon$ (Figure \ref{fig.equilboundaries}, \ref{fig.equil1and3}, \ref{fig.equil1and2}).
\end{enumerate}

Our work also gives a precise way to derive the functional form of \emph{value of privacy} for players (at least in our simple model), in contrast to a large prior literature that has debated the right way to impose exogenously a functional form on player privacy cost functions \cite{GR11,Xiao13,CCKMV13,NVX14}. In contrast to these assumed privacy cost functions, which increase as privacy guarantees are weakened, we show that players can actually sometimes have a negative marginal cost (i.e. a positive marginal utility) for \emph{weakening} the privacy guarantees of a mechanism.

In summary, we show that even in extremely simple models, privacy exhibits much richer behavior in equilibrium compared to its static counterpart, and that decisions about privacy regulation need to take this into account. Our results serve as a call-to-arms --- that especially when making policy decisions about privacy technologies, \emph{equilibrium effects} need to be analyzed, rather than just static effects, because it is otherwise possible that the introduction of a new privacy technology or regulation can have exactly the opposite effect as was intended.

\subsection{Related Literature}
Differential privacy is a quantitative privacy measure, first introduced by  \cite{DMNS06}, that formally quantifies the effect that the behavior of a single individual can have on a signal computed from his behavior. There is a vast literature on algorithms satisfying differential privacy, and we refer readers to \cite{DR14} for a textbook introduction. Differential privacy has been studied in the context of game theory and mechanism design since \cite{MT07}, who showed that it could be used as a \emph{tool} in mechanism design. Subsequently, a moderately sized literature has emerged studying how differential privacy can be used to design mechanisms that incentivize truthful behavior even in settings in which agents have costs for privacy loss --- see e.g. \cite{Xiao13,GR11,CCKMV13,NOS12,NVX14} for a representative but not exhaustive sample. \cite{BMSS15} study a sequential coordination game inspired by financial markets, and show that when players play un-dominated strategies, the game can have substantially higher social welfare when players are given differentially private signals of each other's actions, as compared to when they are given no signal at all. Here, like in other work, however, privacy is viewed as a constraint on the game (i.e. it is added because of its own merits), and does not increase welfare compared to the full information setting. \cite{GL13} study a simple model of data procurement in which user costs for privacy are a function of the number of other users participating --- and study the ability of a mechanism to procure data in equilibrium. In \cite{GL13}, agents have explicitly encoded values for their privacy loss. In contrast, in our model, agents do not care about privacy except insofar as it affects the payoff-relevant outcomes of the game they are playing --- differential privacy in our setting is instead a parameter defining the game we analyze.

A small literature in economics and marketing has looked to understand the effect of privacy in repeated sales settings. The earliest paper is by \cite{taylor2004consumer} who studies a setting where buyers purchase from firm $1$ in period $1$ and firm $2$ in period $2$. The author shows that counter-intuitively, strategic consumers may prefer that their purchase decision be made public, while strategic sellers may prefer to commit to keep purchase decisions private. \cite{calzolari2006optimality} derive similar results in a general contracting setting.

More recently, \cite{conitzer2012hide} consider a setting where a buyer purchases twice from the same firm, and the firm cannot commit to the future price in period $1$, and may condition on the consumer's purchase decision. They consider the effect of allowing the buyer to purchase privacy, i.e. ``hide'' his first period purchase decision from the seller, and show that in equilibrium, having this ability may make buyers worse off (and the firm better off) than in its absence.

We build on these papers by here modeling privacy as a continuous choice variable in the spirit of differential privacy, rather than the discrete choice (purchase decision revealed or not) considered in previous papers. This allows us to analyze the quantitative effect of privacy on welfare and profit as a continuous quantity, rather than a binary effect, and in particular lets us show for the first time that \emph{increasing} privacy protections (i.e. decreasing the correlation between the advertiser's signal and the buyer's action) can actually increase the information contained in the signal about the buyer's type.

\section{Model and Preliminaries}

We study a two period game with a single consumer and two firms.  The first firm has a single good to sell to the consumer and wishes to maximize its expected profit. The second firm is an advertiser who wishes to show a targeted ad to the consumer.  We will also refer to the first firm as the \emph{seller}, and the second firm as the \emph{advertiser}.

\begin{enumerate}
\item In period 1, the consumer has a privately known value $v \in [0,1]$ for the good, drawn from a distribution with CDF $F$ and density $f$.  The seller posts a take-it-or-leave-it price $p$ for the good, and the consumer makes a purchase decision. We assume that the seller's price is not observed by the advertiser. 

\item In period 2, the consumer can have one of two types, $t_1$ and $t_2$, where the probability of having each type depends on her value $v$ from period 1.  Specifically, $\Pr(t_1) = g(v)$ and $\Pr(t_2) = 1-g(v)$, for a known function $g: [0,1] \to [0,1]$. The advertiser may show the buyer one of two ads, $A$ and $B$.  He gets payoff $s_{1A}$ and $s_{2A}$ respectively from showing ad A to a consumer of type $t_1$ and $t_2$, and payoffs $s_{1B}$ and $s_{2B}$ from showing ad B to a consumer of type $t_1$ and $t_2$ respectively.  The consumer gets additional utility $\delta$ from being shown ad A over B.%
\footnote{Since the customer's preferences in period $2$ are independent of her type, it does not matter whether she knows her period $2$ type from the beginning or learns it at the start of period $2$. }
\end{enumerate}

\begin{assumption} \label{ass.one}
The following two assumptions are made regarding the distribution of buyer's value and type:
\begin{enumerate}
\item The distribution of the buyer's value satisfies the non-decreasing hazard rate assumption, i.e. $\frac{f(v)}{1-F(v)}$ is non-decreasing in $v$.
\item The probability that a buyer of value $v$ is of type $t_1$, $g(v)$, is non-decreasing in $v$.
\item Buyers prefer ad A, i.e.  $\delta>0$.
\end{enumerate}
Regarding the payoffs to the advertiser, we assume that $s_{1A}>s_{1B}$ and $s_{2B}>s_{2A}$.
\end{assumption}

The payoff assumption corresponds to the case where type $t_1$ is the ``high type'' and ad A is the ``better'' ad.  For example, the ad could be for a credit card.  The advertiser can offer either a card with a low interest rate (ad $A$) or high interest rate (ad $B$), and the consumer's purchase history may reveal his creditworthiness.  The former distributional assumption is standard in mechanism design. The latter assumption amounts to saying that high-value buyers are more likely to be the ones the advertiser wants to target with the ``good'' ad.

The advertiser neither observes the consumer's type, nor directly observes his purchase decision. Following the consumer's decision, the advertiser will learn (partial) information about the consumer's action in the first period. We write $b$ to denote the bit encoding the consumer's decision in the first period --- that is, $b=1$ if the consumer purchased the good in period 1, and $b=0$ otherwise.  The advertiser does not learn $b$ exactly, but rather a noisy version $\bhat$ that has been flipped with probability $1-q$, for $q \in [1/2,1]$.  The advertiser observes the noisy bit $\bhat$ and then performs a Bayesian update on his beliefs about the consumer's type, and displays the ad that maximizes his (posterior) expected payoff.  The consumer knows that her period 1 purchase decision will affect the ad she sees in period 2.  She seeks to maximize her total utility over both periods, and thus is not myopic.

The parameter $q$ measures the correlation of the reported bit with the actual purchase decision of the consumer, which can also be quantified via differential privacy:

\begin{definition}[\cite{DMNS06}]
A signal $\bhat : \{0,1\}\rightarrow \{0,1\}$ satisfies $\epsilon$-differential privacy if for every $b \in \{0,1\}$ and for every $b' \in \{0,1\}$:
$$\Pr[\bhat(b) = b'] \leq \exp(\epsilon)\Pr[\bhat(1-b)=b']$$
\end{definition}
Differential privacy quantifies in a strong sense what an observer can learn about the input bit from the output --- see \cite{KS14} for an analysis of the semantics of this guarantee --- and among other things provides an upper bound on the mutual information between the input and output bit.

In our setting, it is easy to translate the parameter $q$ into an $\epsilon$-differential privacy guarantee for the Period 1 purchase decision as follows:
\[\frac{q}{1-q}=e^{\epsilon} \; \; \Longleftrightarrow \; \; q = \frac{e^{\epsilon}}{1+e^{\epsilon}}. \]
In the two extremal cases, full privacy ($q=1/2$) corresponds to $\epsilon$-differential privacy with $\epsilon=0$, and no privacy ($q=1$) corresponds to $\epsilon$-differential privacy with $\epsilon=\infty$.  By varying $q$ from $1/2$ to $1$, we will be able to measure changes to the equilibrium outcomes for all possible privacy levels.

\section{Equilibrium analysis}
To begin, we observe that any equilibrium must be such that the consumer follows a cutoff strategy in Period 1: there exists a marginal consumer with value $v^*$, such that any consumer with value $v \in [0, v^*)$ does not buy the good, and any consumer with value $v \in [v^*, 1]$ does. We formally verify this in the Appendix  (Proposition \ref{prop.thresh}).

\subsection{Period 2}
The advertiser sees $\bhat$, the noisy purchase decision bit, and performs a Bayesian update on its prior over types.  This allows us to define the advertiser's posterior given an observed $\bhat = j$, for $j \in \{0,1\}$.  Recall that $\bhat=b$ with probability $q$, and $\bhat=1-b$ with probability $1-q$.

Plugging in the probabilities for each $j \in \{0,1\}$, we see that the advertiser's posterior when he understands the consumer is following a threshold strategy with cutoff $v^*$ is as follows:
\begin{align*} 
r(\bhat=1, v^*, q)  &= \frac{(1-q)\int_{0}^{v^*} g(v) f(v) dv + q \int_{v^*}^{1} g(v) f(v) dv}{(1-q)F(v^*) +q(1-F(v^*))}, \\
 r(\bhat = 0, v^*, q) &= \frac{q \int_{0}^{v^*} g(v) f(v) dv + (1-q) \int_{v^*}^{1} g(v) f(v) dv}{q F(v^*) + (1-q) (1-F(v^*))}.
\end{align*}
Here $r(\bhat, v^*, q)$ is the advertiser's posterior belief that the consumer is of type $t_1$ after seeing the noisy bit $\bhat$, given that $v^*$ is the marginal consumer in Period 1 and given the noise level $q$.

The advertiser wishes to maximize his expected payoff, so his Bayesian optimal decision rule will be to show ad $A$ if and only if,
\[ s_{1A} r(\bhat, v^*,q) + s_{2A} (1-r(\bhat, v^*,q)) > s_{1B} r(\bhat,v^*,q) + s_{2B} (1-r(\bhat, v^*,q)). \]
That is, he will show ad $A$ if it maximizes his expected payoff.  Rearranging this in terms of his posterior, he will show ad $A$ if and only if,
\[ r(\bhat, v^*,q) > \frac{ s_{2B} - s_{2A}}{s_{1A} - s_{2A} - s_{1B} + s_{2B}} := \eta. \]
We define this fraction to be $\eta$ for shorthand.  Notice that $\eta$ does not depend on any game parameters other than the advertiser's payoff for each outcome, $s_{1A}$, $s_{2A},$ $s_{1B}$, and $s_{2B}$.  Further, by our assumptions on the ranking of these four (Assumption \ref{ass.one}) we have ensured that $\eta \in[0,1]$.

The following lemma will be useful throughout. It says that fixing a cutoff strategy $v^*$, the seller's posterior on seeing a noisy ``purchased'' bit is increasing as the amount of noise decreases (i.e. $q$ increases), and similarly the seller's posterior on seeing a noisy ``did not purchase'' bit is decreasing.

\begin{lemma}\label{prop.changeq}
Fixing $v^*$, $r(1,v^*,q)$ is increasing in $q$ and $r(0, v^*,q)$ is decreasing in $q$.
\end{lemma}
\begin{proof}
Define $\alpha_1 = \frac{\int_0^{v^*} g(v) f(v) dv}{F(v^*)}$, $\alpha_2 = \frac{\int_{v^*}^1 g(v) f(v) dv }{1-F(v^*)}$. Note that since $g(\cdot)$ is non-decreasing, $\alpha_1 \leq \alpha_2$. Therefore, $r(1, v^*, q)$ can be written as:
\begin{align*}
r(1,v^*,q) = \frac{(1-q)F(v^*) \alpha_1 + q (1-F(v^*)) \alpha_2}{(1-q)F(v^*) + q(1-F(v^*))}.
\end{align*}
This is the convex combination of $\alpha_1$ and $\alpha_2$ with weights $\frac{(1-q)F(v^*)}{(1-q)F(v^*) + q(1-F(v^*))}$ and $\frac{q(1-F(v^*)) }{(1-q)F(v^*) + q(1-F(v^*))}$ respectively. Next, note that for $q \in [1/2,1]$ the weight on $\alpha_2$ is increasing in $q$, and the weight on $\alpha_1$ correspondingly decreasing. To see this, differentiate the weight with respect to $q$ and observe that it is always positive. Therefore $r(1, v^*, q)$ is increasing in $q.$

Finally, note that $r(0,v^*, q) = r(1, v^*, 1-q)$, so the latter claim follows.
\end{proof}

The following proposition is an important property of the advertiser's posterior, i.e. that in any equilibrium, seeing a noisy ``purchased'' bit always results in a higher assessment of type $t_1$ than a noisy non-purchased bit.

\begin{lemma}\label{prop.morelikely}
For any period 1 cutoff value $v^*$ and any noise level $q$, the advertiser's posterior probability of the consumer having type $t_1$ given noisy bit $\bhat=1$ is higher than his posterior belief of type $t_1$ given noisy bit $\bhat=0$.  Formally, for all $v^*$, $q$, it holds that $r(1,v^*,q) \geq r(0, v^*,q)$.
\end{lemma}
\begin{proof}
The proposition follows from Lemma \ref{prop.changeq}, the fact that $r(0,v^*, q) = r(1, v^*, 1-q)$, and $q \geq 0.5$.
\end{proof}

In light of this, there are only three different strategies that the advertiser could use in equilibrium:

\begin{enumerate}
\item Show ad $A$ to a consumer with noisy bit $\bhat=1$ and ad $B$ to consumer with noisy bit $\bhat=0$.  This is characterized by the following inequalities:
\begin{equation}\label{eq.one} r(1, v^*, q) > \eta \; \; \; \mbox{ and } \;\;\; r(0, v^*,q) < \eta \end{equation}
\item Always shows ad $A$, regardless of the observed noisy bit $\bhat$. This is optimal for the advertiser when the parameters are such that:
\begin{equation}\label{eq.two} r(1, v^*,q) > \eta \; \; \; \mbox{ and } \; \; \; r(0, v^*,q) > \eta \end{equation}
\item Always shows ad $B$, regardless of the observed noisy bit. This is optimal for the advertiser when the parameters are such that:
\begin{equation}\label{eq.three} r(1, v^*,q) < \eta \; \; \; \mbox{ and } \; \; \; r(0, v^*,q) < \eta \end{equation}
\end{enumerate}

In the latter two cases, consumers will behave myopically in the first round because their purchase decision doesn't affect their payoff in the next round.  The seller can then maximize period 1 profits by posting the monopoly price for the distribution $F$.  Thus cases 2 and 3 can only occur when the posterior induced by the monopoly price satisfies \eqref{eq.two} or \eqref{eq.three}.

We call the equilibrium when the advertiser follows the first strategy a \emph{discriminatory advertising equilibrium}. The latter two are referred to as \emph{uniform advertising equilibria} $A$ and $B$ respectively

Define the myopic monopoly price as $p_M,$ i.e. $p_M$ solves:
\begin{align}
p_M - \frac{1-F(p_M)}{f(p_M)} = 0. \label{eqn.pricemono}
\end{align}
The following proposition discusses existence and properties of uniform advertising equilibria.
\begin{proposition}\label{prop.case2and3}
Fixing other parameters of the game:
\begin{enumerate}
\item For $q = \frac12$ there is either a uniform advertising equilibrium $A$ or $B$, but never both.
\item In the former case: uniform advertising equilibria $A$ exist for all $q \in [\frac12, \bar{q}_2]$, where $\bar{q}_2$ is the largest solution to $r(0, p_M, q) = \eta$ in $[\frac12, 1]$, if any. Further, there are no uniform advertising equilibria $B$ for any $q$.
\item In the latter, conversely, uniform advertising equilibria $B$ exist for all $q \in [\frac12, \bar{q}_3]$, where $\bar{q}_3$ is the largest solution to $r(1,p_M, q) =\eta$ in $[\frac12, 1]$, if any. Further there are no uniform advertising equilibria $A$ for any $q$.
\end{enumerate}
\end{proposition}
\begin{proof}
At $q=1/2$, the signal $\bhat$ is complete noise, and the advertiser's posterior on types will be exactly his prior.  This means that the advertiser must show the same ad to all consumers, which corresponds exactly to a uniform advertising equilibrium A or B, depending on whether his prior probability of type $t_1$ is larger or smaller than $\eta$. The corner case where the prior exactly equals $\eta$ is ignored.

Note that $r(1, p_M, \frac12) = r(0,p_M, \frac12)$, and both are either $> \eta$ or $< \eta$. Further, by Lemma \ref{prop.changeq},$r(1, p_M,q)$ is increasing in $q$, while the $r(0, p_M, q)$ is decreasing in $q$.

Therefore the system of equations \eqref{eq.two} or \eqref{eq.three} can only be satisfied on some interval $[\frac12, \bar{q}_j]$ if at all for $v^* = p_M$.
\end{proof}

To collect everything we have shown so far. There are 3 kinds of possible equilibria in this game:
\begin{enumerate}
\item \emph{Discriminatory Equilibrium}: Advertiser shows ad A on seeing a noisy purchase bit and ad B and seeing a noisy non-purchase bit. The cutoffs followed by the consumer, $v^*$ is such that \eqref{eq.one} is satisfied.
\item \emph{Uniform Advertising Equilibrium A}: Advertiser always shows ad A, regardless of bit. In this case buyer purchases myopically, and seller charges the myopic monopoly price $p_M$. Further, $r(\cdot)$ evaluated at $v^*$ equaling the myopic monopoly price $p_M$ satisfies \eqref{eq.two}. This equilibrium exists for all $q$ on some interval $[\frac12, \bar{q}_2],$ if at all.
\item \emph{Uniform Advertising Equilibrium B}: Advertiser always shows ad B, regardless of bit. In this case buyer purchases myopically, and seller charges the myopic monopoly price $p_M$. Further, $r(\cdot)$ evaluated at $v^*$ equaling the myopic monopoly price $p_M$ satisfies \eqref{eq.three}. This equilibrium exists for all $q$ on some interval $[\frac12, \bar{q}_3],$ if at all.
 \end{enumerate}
By observation, the two types of uniform advertising equilibria cannot coexist in the same game.

In a uniform advertising equilibrium, the period 1 behavior is straightforward.  We now finish the analysis by characterizing period 1 behavior under discriminatory advertising equilibria.

\subsection{Period 1 Behavior in Discriminatory Advertising Equilibria}
In this kind of equilibrium the consumer is aware that her period 1 purchasing decisions will affect the ad she sees in period 2.  She will buy in period 1 if and only if her surplus from purchasing at price $p$ plus her continuation payoff from having purchased (i.e.  expected utility from the ad that will be shown) is greater than the continuation payoff from not having purchased.   Formally, a consumer with value $v$ will purchase in period 1 if the following holds:
\begin{align} \label{eqn:purchase}
(v - p) + q \delta & \geq (1-q) \delta.
\end{align}
For the marginal consumer with value $v^*$, this inequality must hold with equality.
\[ v^* = p + (1-2q)\delta. \]

Define $p_1(q)$ to be the seller's optimal price charged at noise level $q$. Further define $v^*(q)$ as the implied cutoff type at noise level $q$, i.e. $v^*(q) = p_1(q) + (1-2q) \delta.$ Note that both $p_1(q)$ and $v^*(q)$ are continuous functions of $q$.

\begin{lemma}\label{prop.deriv}
Assuming a discriminatory advertising equilibrium exists for a neighborhood of $q \in [\frac12,1]$, the optimal price $p_1(q)$ is increasing in $q$ while the cutoff type $v^*(q)$ is decreasing.
\end{lemma}
\begin{proof}
From equation \eqref{eqn:purchase} above, if the seller charges a price of $p$ in a discriminatory advertising equilibrium, then the buyer purchases if her value exceeds $p + (1-2q) \delta$. Therefore the seller chooses $p$ to maximize his net profit,
\begin{align*}
p (1-F(p + (1-2q) \delta)).
\end{align*}
Differentiating with respect to $p$, the optimal price ($p_1(q)$) in a discriminatory advertising equilibrium solves:
\begin{equation}\label{eqn:foc}
p_1(q) - I(p_1(q) + (1-2q)\delta) =0,
\end{equation}
where $I(v) = \frac{1-F(v)}{f(v)}.$ Applying the implicit function theorem, we see that
\begin{equation*}
p_1'(q) - I'(p_1 + (1-2q)\delta) (p_1'(q) - 2 \delta) =0.
\end{equation*}
Therefore, since $I'$ is negative, it follows that $p_1'(\cdot)$ must be positive. Further, we have that $p_1'(q) - 2\delta$ is negative, i.e. the cutoff value who buys at the optimal price charged is decreasing in $q$.
\end{proof}

The following proposition says that whenever a discriminatory advertising equilibrium exists at a given $q$, the price is higher and more customers buy than in either uniform advertising equilibrium.
\begin{lemma}\label{prop.ordering}
For any $q \in [1/2, 1]$, the period 1 price in a discriminatory advertising equilibrium (if it exists), is higher than the monopoly price (i.e. the price charged in a uniform advertising equilibrium), which is higher than the purchase cutoff employed by a consumer in a discriminatory advertising equilibrium. Formally, $$v^*(q) \leq p_M \leq p_1(q).$$
\end{lemma}
\begin{proof}
Note that at $q = \frac12$, $p_1(q) = v^*(q) = p_M.$
The result now follows since $p_1'(\cdot)$ is positive, while $v^{*\prime}(\cdot)$ is negative.
\end{proof}

Finally, to resolve existence, which follows easily from the definitions.
\begin{observation}\label{obs.existence}
A discriminatory advertising equilibrium exists at all $q$ such that $r(1, v^*(q), q) \geq \eta$ and $r(0, v^*(q), q) \leq \eta$.
\end{observation}
Next, note that $v^*(q)$ is a continuous function of $q$. Therefore, equilibria of type $1$ may exist for possibly multiple disjoint intervals in $(\frac12,1]$.

\section{Illustrations via an Example}\label{s.example}

In this section we highlight some of the counter-intuitive effects that result from changing the level of differential privacy constraining the advertiser's signal, by means of an explicit family of simple examples. The phenomena we highlight are quite general, and are generally not brittle to the choice of specific parameters in the game. For simplicity of exposition, we highlight each of these phenomena in the simplest example in which they arise. For the remainder of this section, we take the distribution of buyer values, $F$, to be the uniform distribution on $[0,1]$. We also set the buyer's probability of having type $t_1$ to be exactly his value ---  i.e. we take  $g(v)=v$ for all $v \in [0,1]$.  Finally, we set the additional utility that a buyer gets from being shown ad A to be $\delta = 1$.  The value of parameter $\eta$ (along with its defining parameters $s_{1A}$, $s_{1B}$, $s_{2A}$, and $s_{2B}$) will vary by example, and will be specified when relevant.

The static monopoly price in this game is $p_M = 1/2$, and the price and cutoff value in a discriminatory equilibrium (if it exists at a given $q$) are
\[ p_1(q) = \frac{1}{2} + (2q-1)\frac{\delta}{2} \; \; \mbox{ and } \; \; v^*(q) = \frac{1}{2} - (2q-1)\frac{\delta}{2}. \]

Below we plot these values as a function of $q$.  Note that in this particular example, the discriminatory price and cutoff value are linear in $q$ (because values are distributed uniformly), although this need not be the case in general.

\begin{figure}[h!]
\centering
\includegraphics[scale=0.4]{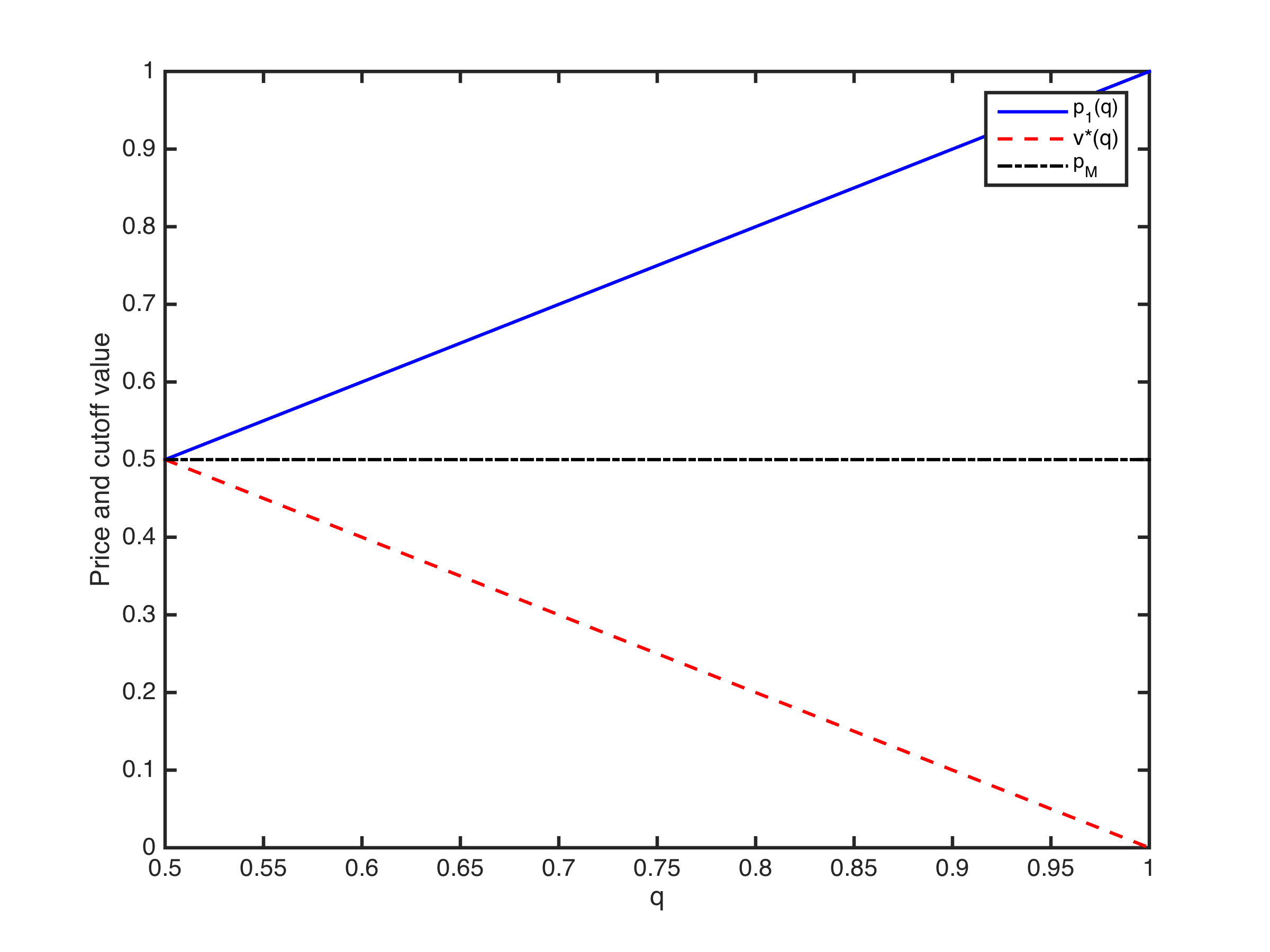}
\caption{A plot of the monopoly price $p_M$ and the equilibrium discriminatory price $p_1$ and cutoff value $v^*$ as a function of $q$.  In a uniform equilibrium, the cutoff value is equal to the monopoly price because consumers behave myopically.}
\label{fig.priceandcutoff}
\end{figure}

Figure \ref{fig.priceandcutoff} shows that as $q$ increases, the discriminatory equilibrium price increases and cutoff value (value of marginal consumer that purchases in period 1) decreases.  Relative to a uniform advertising equilibrium, more consumers purchase the good in period 1, and at a higher price in a discriminatory equilibrium.  Observe further that at $q=1$ (i.e. no privacy), all consumers purchase in period 1, regardless of their value. This is because the value in period 2 of hiding information relevant to their type exceeds the loss that they take by buying at a loss in period 1. Here, when the consumers are offered no privacy protections, they in effect change their behavior to guarantee their own privacy.

The existence and types of equilibria in this game depend on the advertiser's posterior, given that the prices and cutoff values above will arise in period 1.  The advertiser's posterior beliefs about the consumer's type for each realization of $\bhat$ are given below for both the discriminatory and uniform advertising equilibria.
\begin{align*}
r(1,v^*(q),q) &= \frac{-2q^3 + 5q^2 - 3q + 1}{4(q^2 - q + \frac{1}{2})} \\
r(0,v^*(q),q) &= \frac{-2q^3 + 5q^2 - 3q}{4(q^2 - q)} =  \frac{3-2q}{4} \\
r(1, p_M,q) &= \frac{1+2q}{4} \\
r(0, p_M,q) &= \frac{3-2q}{4}
\end{align*}

To illustrate the existence of various equilibria, we plot the advertiser's possible posterior beliefs below as a function of $q$.

\begin{figure}[h!]
\centering
\includegraphics[scale=0.4]{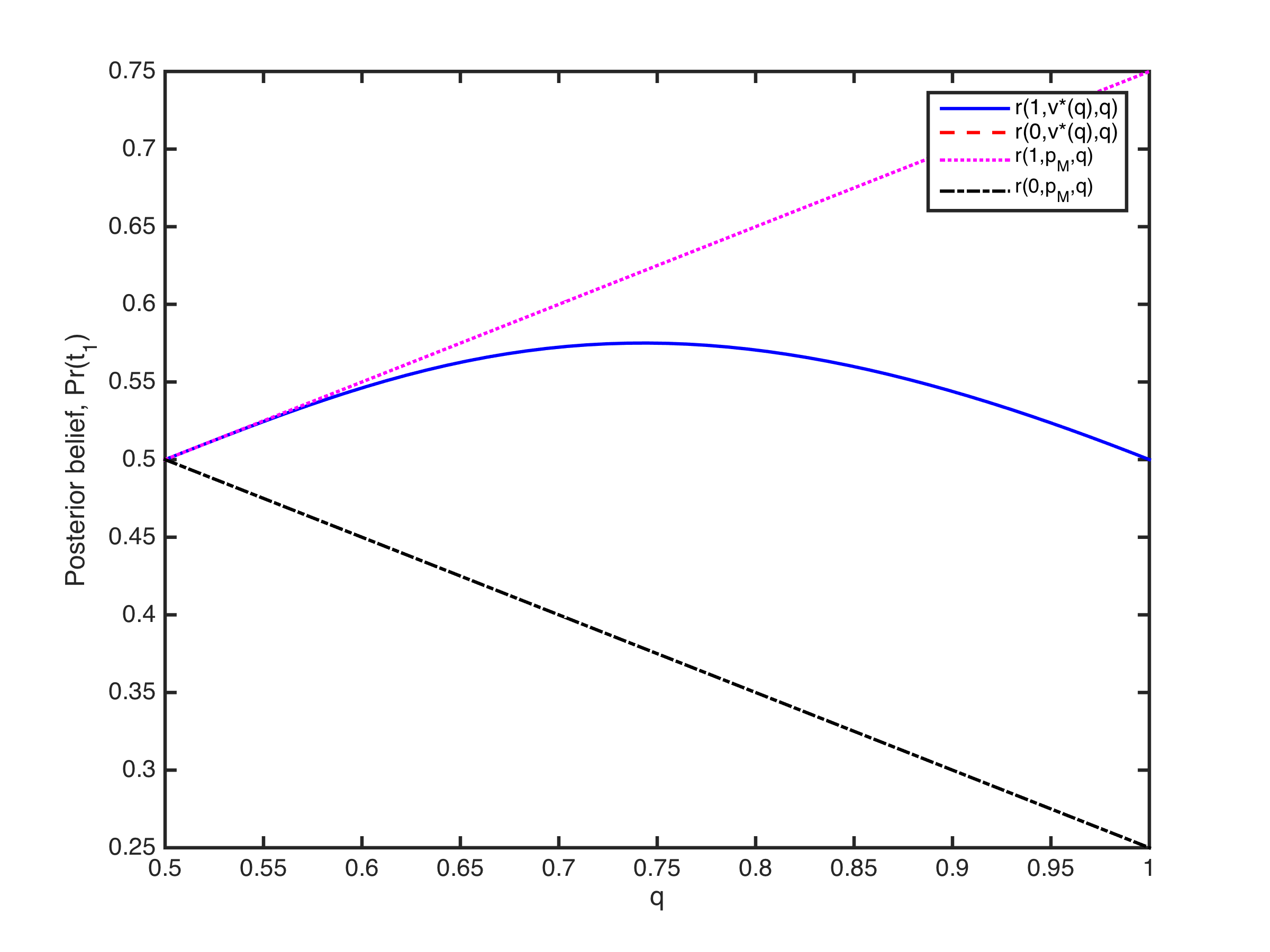}
\includegraphics[scale=0.4]{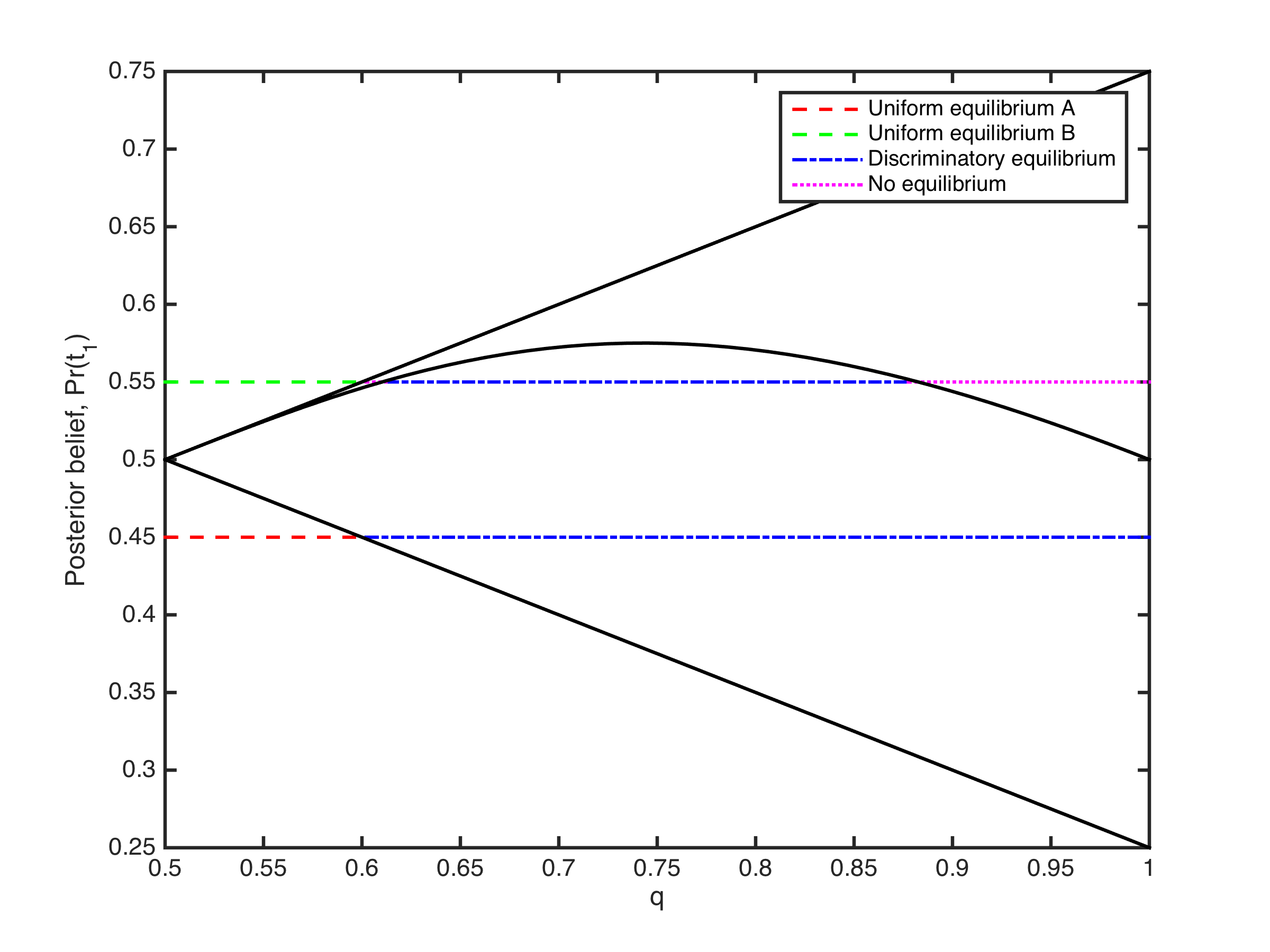}
\caption{On the left is a plot of the advertiser's possible posteriors in both discriminatory and uniform equilibria, given the corresponding prices and cutoff values of Period 1 and an observation of $\bhat$.  In this example, $r(0,v^*(q),q)=r(0, p_M,q)$ for all $q$.  On the right, we show for $\eta=0.45$ and $\eta=0.55$ how these posteriors correspond to different equilibrium types as $q$ varies.}
\label{fig.equilboundaries}
\end{figure}

Note the following counter-intuitive fact: the advertiser's posterior, having seen a noisy ``purchased'' bit, i.e., $\bhat=1$, in a discriminatory equilibrium is non-monotone in the noise level $q$. Statically, if we were to increase the privacy level (i.e. decrease q), we should always expect the advertiser's posterior to be \emph{less informative} about the consumer's type, but as we see here, in equilibrium, increasing the privacy level can sometimes make the advertiser's posterior more accurate.  This can occur because of the two competing implications of adding less noise (i.e. increasing $q$):  on the one hand, the observed bit $\bhat$ is less noisy, and is thus a more accurate indicator of the consumer's purchase decision in period 1.  On the other hand, as $q$ increases, a larger fraction of consumers buy in period 1; the pool of consumers who do purchase the item is ``watered down'' by low-valued consumers who are unlikely to have type $t_1$. This can be viewed as a larger fraction of consumers modifying their behavior to guarantee their own privacy, as the privacy protections inherent in the market are weakened.  The dominating effect on the posterior depends on $q$ and the game parameters.

Similar non-monotonicities can never be observed when the advertiser sees $\bhat=0$ in a discriminatory equilibrium because these two implications are no longer at odds.  As we reduce the amount of noise added, the noisy bit $\bhat$ is still more likely to be accurate.  In addition, the maximum value $v^*(q)$ of a consumer who did not purchase is decreasing in $q$, ensuring that only the lowest valued (and thus the least likely to have type $t_1$) consumers do not purchase in period 1.  These two effects conspire to ensure that $r(0,v^*(q),q)$ is monotonically decreasing in $q$.

Figure \ref{fig.equilboundaries} can be also be used to illustrate to existence of equilibria in this family of games.  Specifying a value of $\eta$ for the game determines the type of equilibria (if any) that exist at each $q$, according to Conditions \eqref{eq.one}, \eqref{eq.two}, and \eqref{eq.three}.  This can be easily visualized using Figure \ref{fig.equilboundaries}.  

Equilibria need not exist for all ranges of $q$: it is possible for none of Conditions \eqref{eq.one}, \eqref{eq.two}, or \eqref{eq.three} to be satisfied for a given $\eta$ and $q$.  
However, in all games, there is a uniform equilibrium at $q=1/2$, where consumers behave myopically in period 1 and then no information is shared with the advertiser.  Equilibria also need not be unique for a given $q$; a discussion of equilibrium multiplicity is deferred to Section \ref{s.mult}.

Next we show that in settings for which a discriminatory equilibrium exists for a range of $q$, the mutual information between the noisy bit $\bhat$ and the consumer's type can be non-monotone in $q$.  In particular, as we increase the privacy protections of the market (i.e. decrease $q$), we can sometimes end up \emph{increasing} the amount of information about the consumer's type present in the advertiser's signal!  This is for similar reasons to those that lead to non-monotonicity of the advertiser's posterior belief --- as the market's privacy protections decrease, consumers change their behavior in order to guarantee their own privacy.   Figure \ref{fig.mutualinfo} plots the mutual information between $\bhat$ and the consumer's type as a function of $q$ in a discriminatory equilibrium.\footnote{As illustrated by Figure \ref{fig.equilboundaries}, there are games for which a discriminatory equilibria exist for the relevant range of $q$, e.g., when $\eta= 1/2$, a discriminatory equilibrium exists for all $q \in [1/2,1]$.} Note that although both the advertiser's posterior and the mutual information between the consumer's purchase decision and the signal exhibit similar non-monotonicities in $q$, they do not peak at the same value of $q$!

\begin{figure}[h!]
\centering
\includegraphics[scale=0.4]{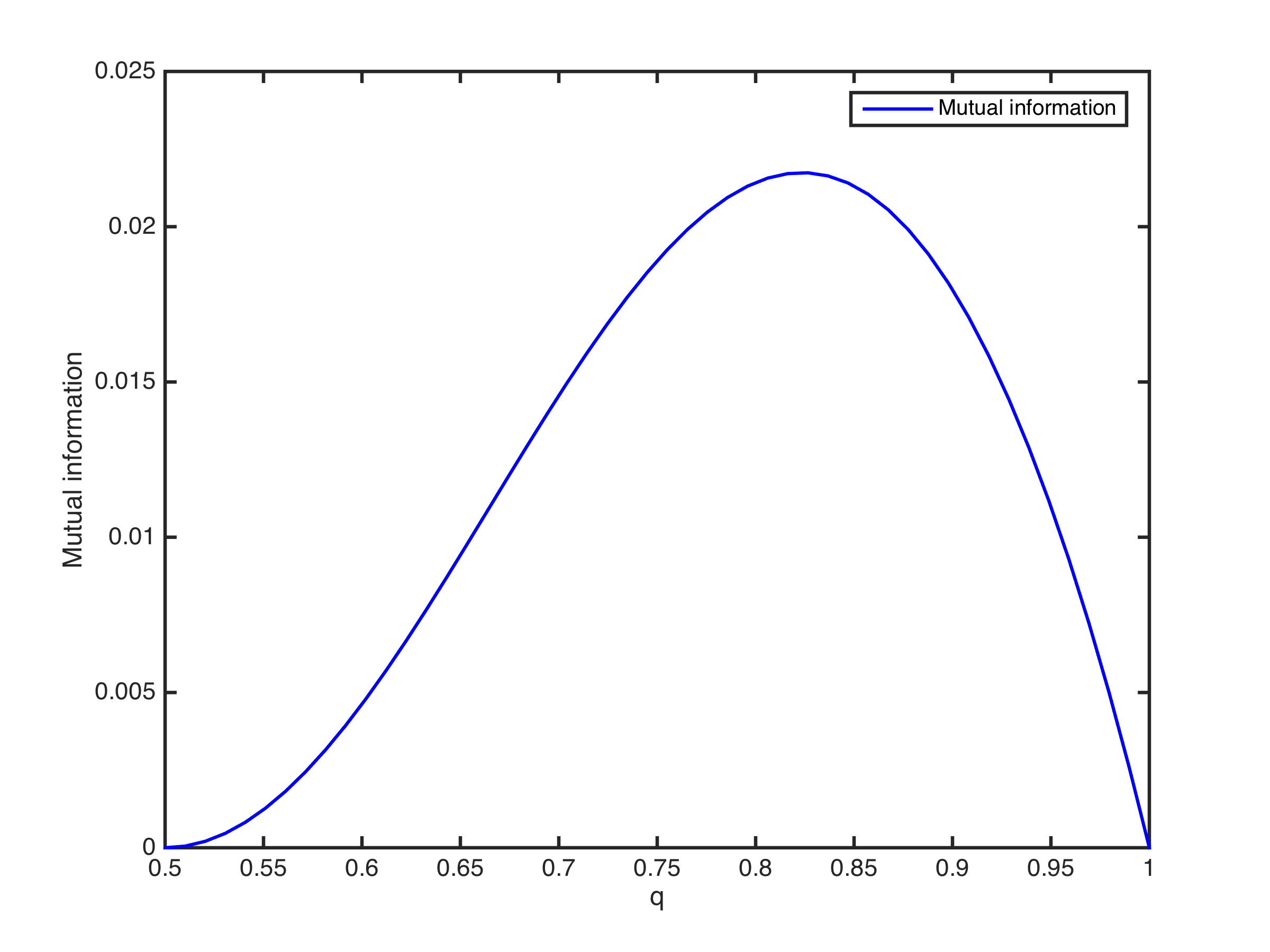}
\caption{A plot of the mutual information between $\bhat$ and the consumer's type in a discriminatory equilibrium.}
\label{fig.mutualinfo}
\end{figure}

These phenomena together suggest that in the range of $q$ in which the mutual information is decreasing, the advertiser might actually prefer that the market include stronger privacy guarantees.  Indeed, Figure \ref{fig.advertiserutility} plots the advertiser's utility in a discriminatory equilibrium as a function of $q$, where $s_{1A}=s_{2B}=1$ and $s_{1B}=s_{2A}=0$.  This setting of parameters gives $\eta=1/2$, where a discriminatory equilibrium exists for all $q$.

\begin{figure}[h!]
\centering
\includegraphics[scale=0.4]{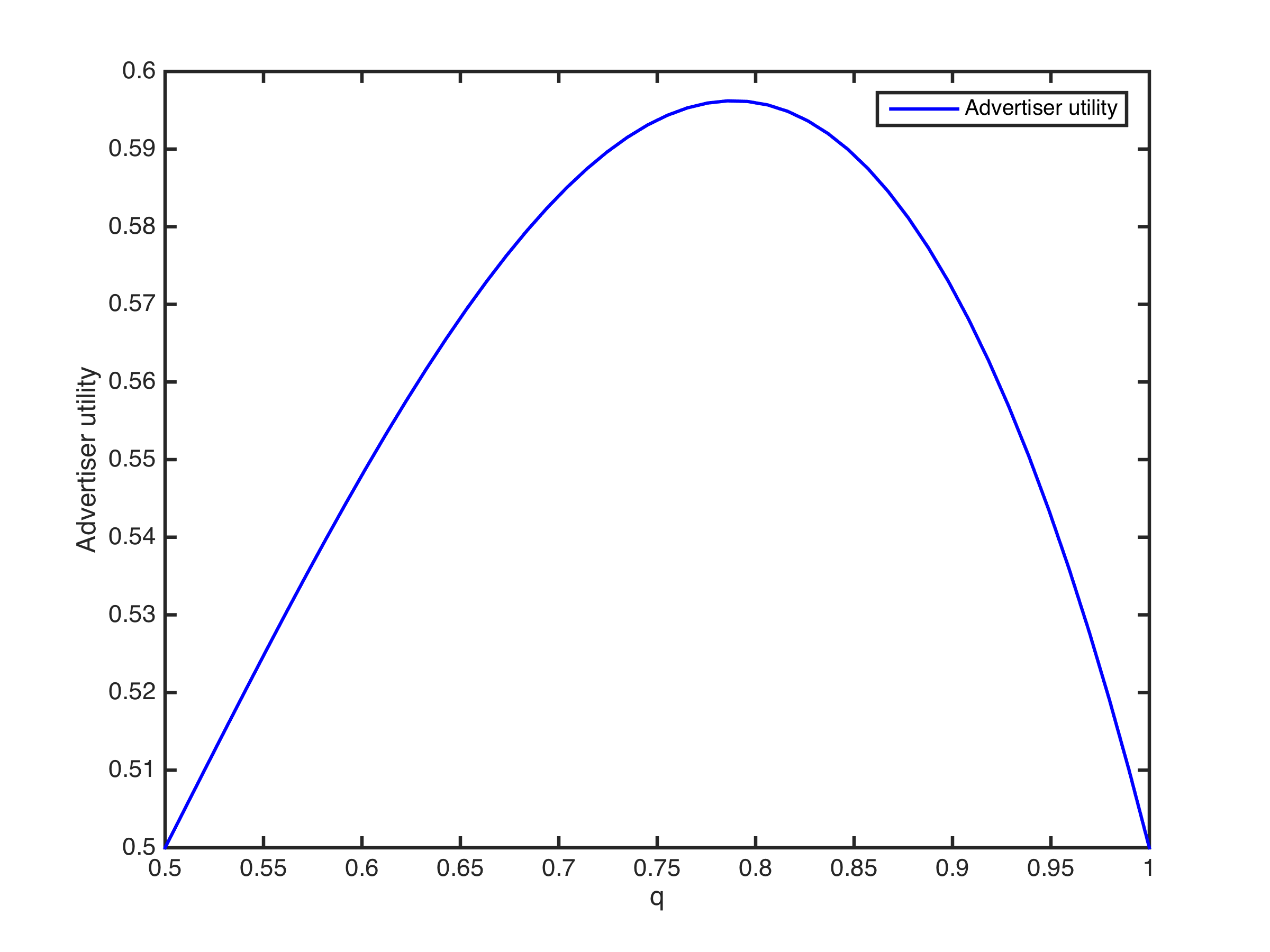}
\caption{A plot of the advertiser's utility in a discriminatory equilibrium, where $s_{1A}=s_{2B}=1$ and $s_{1B}=s_{2A}=0$. These parameters correspond to $\eta = 1/2$, which ensures that a discriminatory equilibrium exists for all $q \in [1/2,1]$.}
\label{fig.advertiserutility}
\end{figure}

As predicted, the advertiser's utility is non-monotone in $q$.  However, the advertiser's utility is not maximized at the same value of $q$ that maximizes the mutual information.  Thus, the advertiser's interests are not necessarily aligned with the goal of learning as much information about the consumer as possible --- and are certainly not incompatible with privacy technologies being introduced into the market. Indeed, the ideal level of $q$ for the advertiser is strictly on the interior of the feasible set $[1/2,1]$.

We would also like to understand how consumer surplus and profit vary with $q$. These are confounded by the fact that for a fixed level of $\eta$, an equilibrium type may or may not exist for a given $q$. To simplify the analysis to not account for this existence problem, consider the following though experiment: For each equilibrium type, and any $q$, pick $\eta$ such that the appropriate one of Conditions \eqref{eq.one}, \eqref{eq.two}, or \eqref{eq.three}  is satisfied. Fixing the advertiser's equilibrium behavior, $\eta$ only affects the advertiser's payoff, not the seller's or buyer's. Figure \ref{fig.csandprofit} plots the consumer's surplus and the seller's profit as a function of $q$ for each equilibrium type, under this artificial thought experiment.

\begin{figure}[h!]
\centering
\includegraphics[scale=0.4]{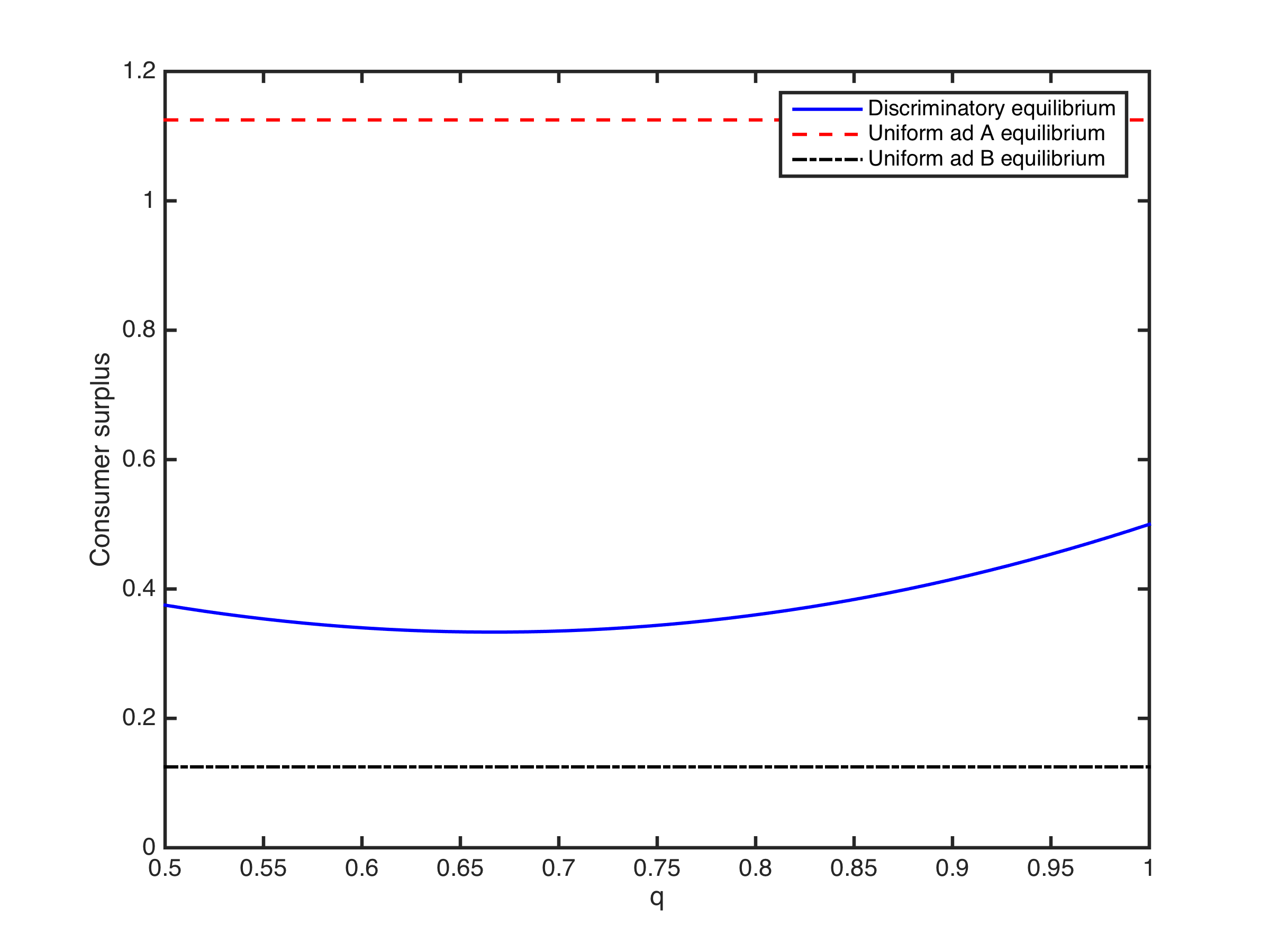}
\includegraphics[scale=0.4]{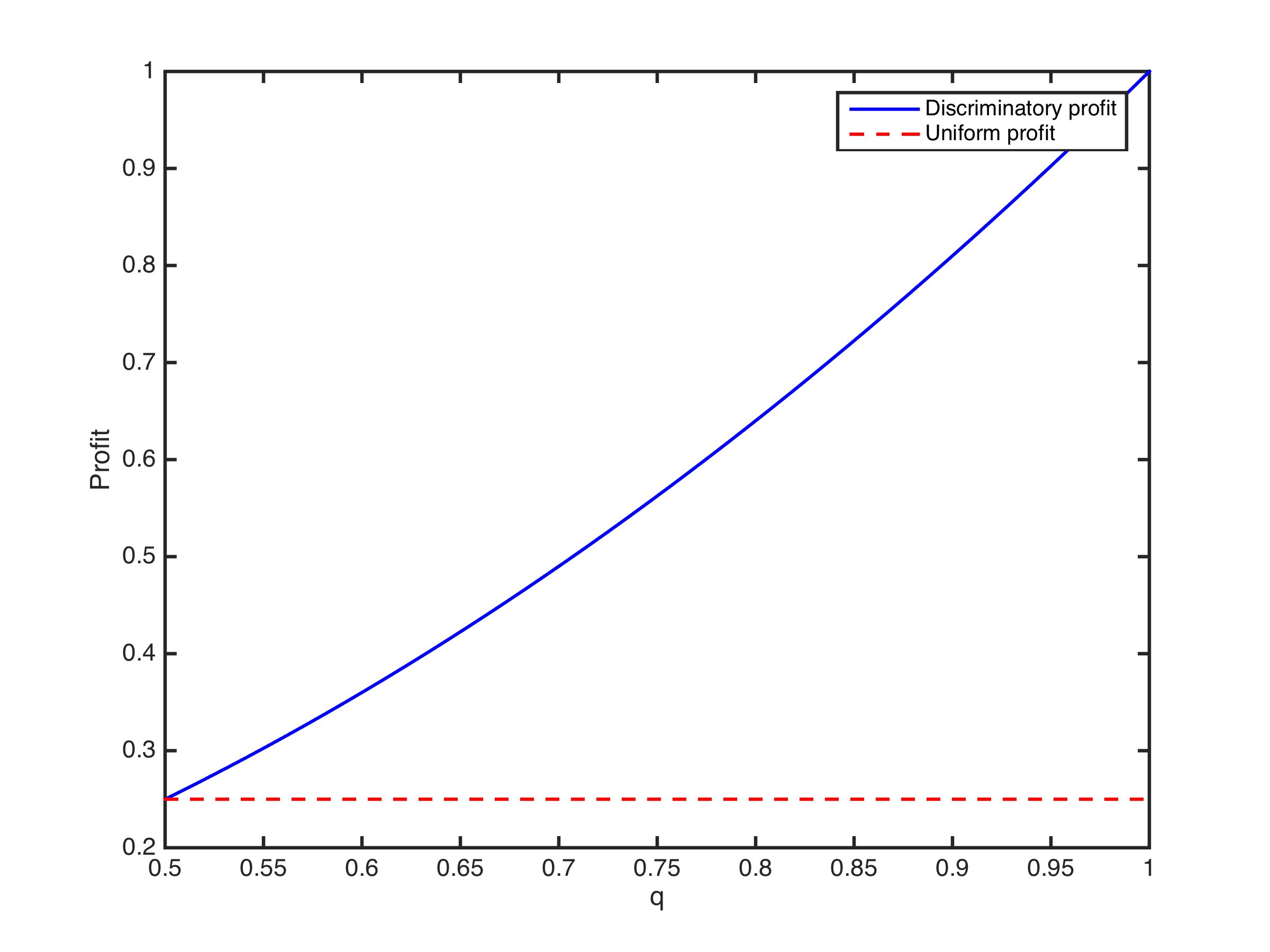}
\caption{On the left is a plot of consumer surplus in all three equilibrium types.  On the right is a plot of the seller's profit in both discriminatory and uniform equilibria.  The seller's profit in a uniform equilibrium doesn't depend on the ad shown in period 2.}
\label{fig.csandprofit}
\end{figure}

In a discriminatory equilibrium, the consumer surplus is convex in $q$; and although at certain points on the curve, consumers have positive marginal utility for increased privacy (i.e. a smaller value of $q$), the consumer's globally optimal value is $q=1$.  This means that given the equilibrium effects, consumers fare the best with no privacy and would prefer to have their purchase decision revealed exactly.  The seller also prefers $q=1$, because as $q$ increases, more consumers purchase the good at a higher price in equilibrium.  Unsurprisingly, consumer surplus and revenue are constant across $q$ in uniform advertising equilibria.

In addition to changes for a fixed equilibrium type, varying $q$ can also change the type of equilibria that exist in the game.  Figure \ref{fig.equilboundaries} showed that small changes in $q$ can cause new equilibria to spring into existence or disappear.  Thus the consumer surplus, profit, and advertiser's utility can jump discontinuously in $q$.  We illustrate this phenomenon with two examples: the first example illustrates changes between a uniform advertising equilibrium B, a discriminatory equilibrium, and no equilibrium.  The second shows a change from a uniform advertising equilibrium A to a discriminatory equilibrium.

For the first example, set $s_{1A}=.5$, $s_{2B}=.6$, and $s_{1B}=s_{2A}=.05$, which implies $\eta = .55$.  As illustrated in Figure \ref{fig.equilboundaries}, as we increase $q$ from $1/2$ to $1$, there is first a uniform advertising equilibrium B, then equilibria briefly cease to exist, then discriminatory equilibria exist, and finally no equilibria exist for large $q$.  Each change of equilibrium existence results in discrete jumps in the consumer surplus, profit, and advertiser utility.  Figure \ref{fig.equil1and3} illustrates below.

\begin{figure}[h!]
\centering
\includegraphics[scale=0.4]{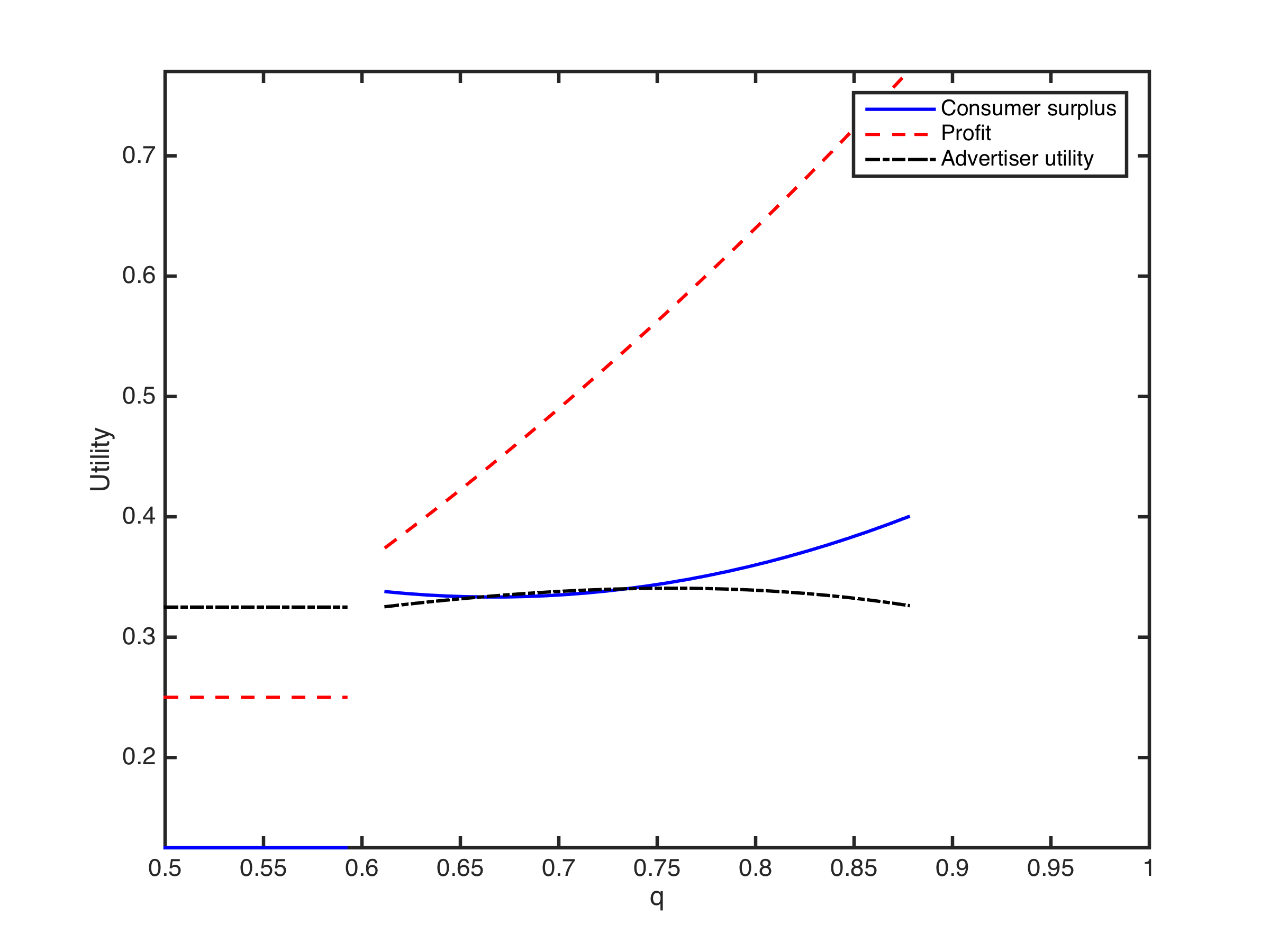}
\caption{A plot of consumer surplus, profit, and advertiser utility in equilibrium when $s_{1A}=.5$, $s_{2B}=.6$, and $s_{1B}=s_{2A}=.05$.  Discontinuities correspond to changes of the equilibrium type that exists, or an absence of equilibria altogether.}
\label{fig.equil1and3}
\end{figure}

For the second example, set $s_{1A}=.6$, $s_{2B}=.5$, and $s_{1B}=s_{2A}=.05$, which implies $\eta = .45$.  In this game, as $q$ increases from $1/2$ to $1$, the equilibrium type changes discretely from a uniform advertising equilibrium A to a discriminatory equilibrium, also illustrated in Figure \ref{fig.equilboundaries}.  This change also causes discontinuities in the consumer surplus, profit, and advertiser utility. Note here that a tiny decrease in the level of privacy (i.e. increase in $q$) can cause a precipitous drop in welfare. Figure \ref{fig.equil1and2} illustrates this effect.

\begin{figure}[h!]
\centering
\includegraphics[scale=0.4]{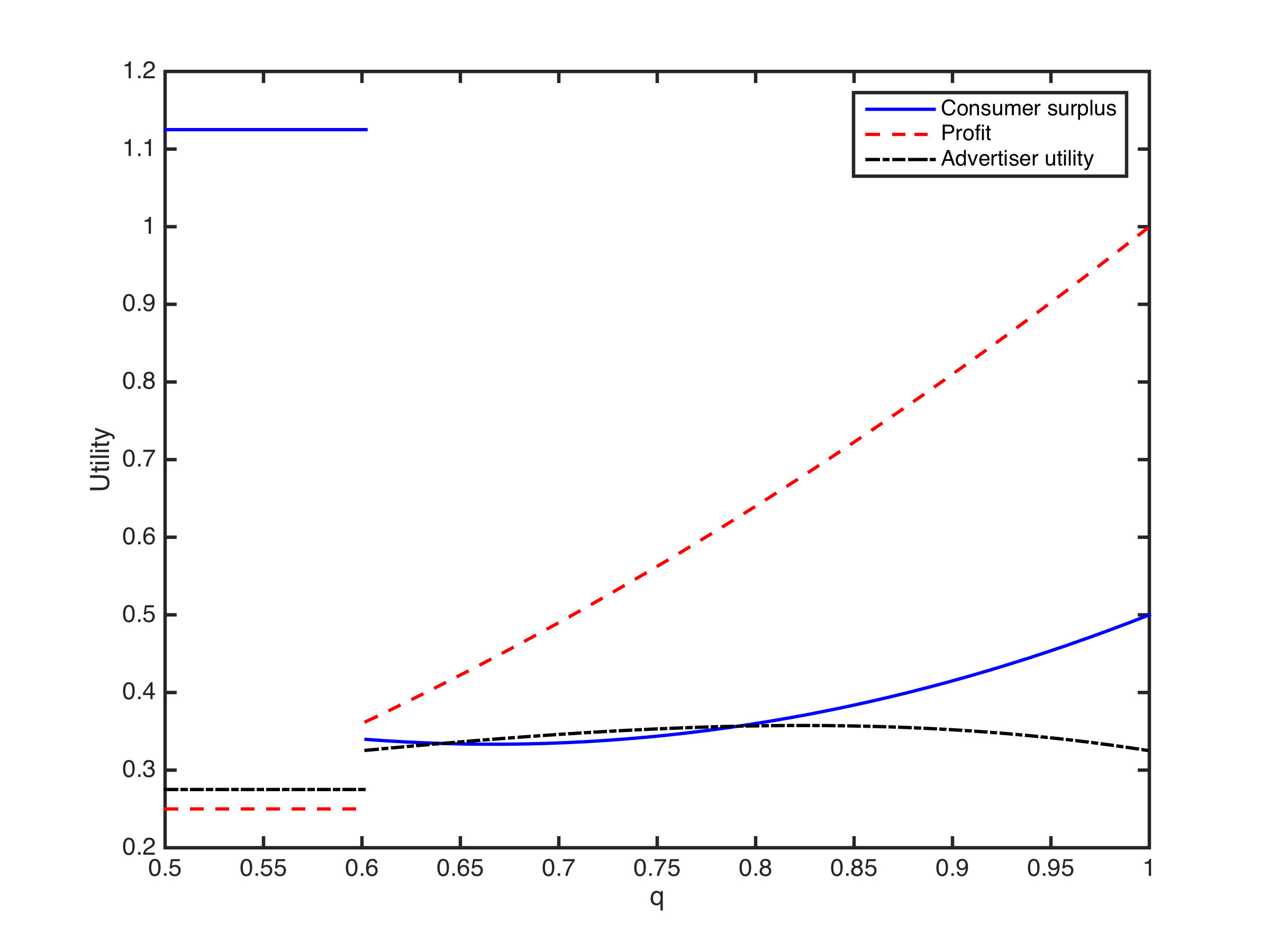}
\caption{A plot of consumer surplus, profit, and advertiser utility in equilibrium when $s_{1A}=.6$, $s_{2B}=.5$, and $s_{1B}=s_{2A}=.05$.  Discontinuities correspond to changes of the equilibrium type that exists.}
\label{fig.equil1and2}
\end{figure}

\section{Multiplicty of Equilibria and Other Results}\label{s.mult}
In this section we build on the example of the previous section to provide some formal results about equilibrium multiplicity, welfare, etc. as a function of the promised level of privacy offered to the consumer. Proofs are deferred to the Appendix.

\subsection{Multiplicity}
\begin{restatable}{proposition}{eqoneandtwo}\label{prop.eq1and2}
A discriminatory equilibrium and a uniform advertising equilibrium A can coexist in the same game for the same level of $q$.
\end{restatable}

\begin{restatable}{proposition}{cwcaseoneandtwo}\label{prop.cwcase12}
If a discriminatory equilibrium and a uniform advertising equilibrium A coexist at a given noise level $q$, then the buyer prefers the uniform equilibrium to the discriminatory equilibrium, regardless of her value $v$.  On the other hand, the seller prefers the discriminatory equilibrium.
\end{restatable}

Proposition \ref{prop.cwcase12} is intuitive, so we omit a formal proof. To see the first claim, observe that a buyer faces both a lower price in period 1 and sees a better ad in period 2, so she is always better off in the uniform advertising equilibrium, regardless of her value. To see the latter, observe that the discriminatory equilibrium allows the seller to sell to more consumers ($v^*(q) < p_M$) at a higher price ($p_1(q) > p_M$).

\paragraph{Takeaway} Small changes in $q$ can make uniform advertising equilibria cease to exist, 
 and can therefore have a discrete impact on welfare and revenue, as they cause the equilibrium to shift discontinuously from uniform to discriminatory. At these boundaries, the buyer may (strictly, discontinuously) prefer slightly less privacy while the seller may (strictly, discontinuously) prefer more privacy!

\begin{restatable}{proposition}{eqoneandthree}\label{prop.1and3}
A discriminatory equilibrium and a uniform advertising equilibrium B can coexist in the same game for the same level of $q$.
\end{restatable}

\begin{restatable}{proposition}{cwcaseoneandthree}\label{prop.cwcase13}
In any game where both a discriminatory equilibrium and a uniform advertising equilibrium B exist for the same value of $q$, average consumer welfare is always higher under the discriminatory equilibrium, but individual consumers may have different preferences for these two equilibria.
\end{restatable}

\paragraph{Takeaway} In a game where a discriminatory equilibrium and uniform advertising equilibrium B coexist, buyers prefer the discriminatory equilibrium. Since the uniform advertising equilibrium B exists for an interval of ``low'' $q$, a buyer therefore may prefer less privacy!

\subsection{Welfare Comparative Statics and Preferences over Levels of Privacy}

\begin{restatable}{proposition}{cwcaseone}\label{prop.cwcase1}
In settings where a discriminatory equilibrium exists for a range of $q$, the equilibrium (consumer) welfare can be increasing in $q$.
\end{restatable}

\paragraph{Takeaway} Note that the example in the previous section already demonstrates that this can be true on some interval of $q$, see for example Figure \ref{fig.csandprofit}. The proof in the appendix demonstrates this formally. Note that this implies the buyer's preferences over different levels of privacy may be complex, and her wefare is not necessarily the monotonically increasing in her privacy as is often assumed.

\begin{observation}
The seller prefers the largest $q$ consistent with discriminatory equilibrium --- he gets a higher price and more demand. If a discriminatory equilibrium ceases to exist at some interior $q$, the seller will at that point prefer that the market provide \emph{more} privacy.
\end{observation}

\paragraph{Takeaway} The seller prefers discriminatory equilibria over uniform advertising equilibria, and the least privacy that is consistent with a discriminatory equilibrium if given the choice over privacy levels.
 
\begin{observation} 
For any level of noise $q$, the advertiser always prefers a discriminatory equilibrium to a uniform equilibrium if both exist.
\end{observation}

To see this, note that in any uniform equilibrium, the advertiser's net utility is the same as his ex-ante utility from showing that ad (since he ignores information from period 1). Since he chooses to act on the information he gets in a discriminatory equilibrium, his net utility must exceed the ex-ante utility of showing the same ad. 

\paragraph{Takeaway} The advertiser will always prefer levels of $q$ consistent with discriminatory equilibria. Among these, however, he may prefer strictly interior levels of $q$, for example, as demonstrated by Figure \ref{fig.advertiserutility} previously.

\section{Concluding Remarks}

A rich body of work on differential privacy has developed in the computer science literature over the past decade. Broadly caricatured, this literature provides algorithms for accurate data analyses (of various sorts), subject to guaranteeing $\epsilon$-differential privacy to individual entries in the dataset. Typically, $\epsilon$ can be set to be any value (implicitly to be chosen by the entity that controls access to the dataset), and mediates a trade-off between the strength of the privacy guarantee offered to the data owners and the accuracy of the analysis promised to the user of the data. This trade-off is typically viewed in simple terms: higher values of $\epsilon$ (i.e. less privacy) are better for the analyst because they allow for higher accuracy, the reasoning goes, and worse for the privacy-desiring individual. Of course, if the dataset is already gathered, this reasoning is correct.

In this paper, we add some caveats to this folk wisdom in a simple stylized model. The dataset here consists of the individuals' purchase decision of a good. The data analyst in our model is an advertiser. Individuals do not care about the privacy of their purchase decision for its own sake, but rather, care about privacy only insofar as it affects how ads are targeted at them. A crucial point is that in our model, at the time of choosing the privacy policy, these purchase decisions have not yet been made. As a result, the price of the good, the purchase decision of the individual, and the advertising policy of the advertiser all depend on the announced privacy policy. Evaluations of the privacy preferences of seller, advertiser, and buyers must take into account everyone's equilibrium incentives. As we demonstrated, these can be the opposite of the simple static trade-offs we are used to, and reasoning about them correctly can be complex.

As the literature expands from (the already hard) questions of privately analyzing existing datasets, to thinking of setting privacy policies that influence future user behavior and the datasets that result from this behavior, the equilibrium approach we espouse here will be important. We hope this paper serves as a call-to-arms to reasoning about privacy policy in such settings, while also highlighting the difficulties.

\bibliographystyle{plainnat}
\bibliography{repeated}

\newpage

\appendix

\section{Appendix}

We now show that the only equilibria of this game have a cutoff strategy in period 1. 

\begin{proposition}\label{prop.thresh}
All equilibria have the property that in period 1, there exists a threshold value $v^*$ such that the consumer buys if and only if $v> v^*$.
\end{proposition}
\begin{proof}
Assume not.  Then there exists $v, v'$ such that $v'<v$, and in equilibrium consumers with value $v'$ buy in period 1, while consumers with value $v$ do not.  We consider the three possible equilibrium types, and show that each one leads to a contradiction.

If this is a discriminatory equilibrium, then consumers with $\bhat=1$ are shown ad $A$, and consumers with $\bhat=0$ are shown ad $B$.  Then it must be the case that the consumer's utility satisfies:
\[u(v', \text{buy}) \geq u(v', \text{not buy}) \; \; \; \mbox{ and } \; \; \; u(v, \text{not buy}) \geq u(v, \text{buy})  \]
This implies:
\begin{align*}
& u(v', \text{buy}) + u(v, \text{not buy}) \geq u(v', \text{not buy}) + u(v, \text{buy}) \\
\Longleftrightarrow \; \; \; & \left[ v' - p + q \delta \right] + \left[ (1-q) \delta \right] \geq \left[ (1-q) \delta \right] + \left[ v - p + q \delta \right] \\
\Longleftrightarrow \; \; \; & v' \geq v
\end{align*}
This is a contradiction because $v > v'$.

If this is a uniform advertising equilibrium A, then all consumers are shown ad $A$ and receive utility $\delta$ in period 2.  It must be the case that:
\begin{align*}
& u(v', \text{buy}) \geq u(v', \text{not buy})\\
\Longleftrightarrow \; \; \; & v'-p+\delta \geq \delta \\
\Longleftrightarrow \; \; \; & v' \geq p
\end{align*}
Also,
\begin{align*}
& u(v, \text{not buy}) \geq u(v, \text{buy})\\
\Longleftrightarrow \; \; \; & \delta \geq v-p + \va \\
\Longleftrightarrow \; \; \; & p \geq v
\end{align*}
These two facts above imply that $v' \geq v$, which is a contradiction since we have assumed that $v>v'$.

If this is a uniform advertising equilibrium B, then all consumers are shown ad $B$ and receive zero utility in period 2.
\[ u(v', \text{buy}) \geq u(v', \text{not buy}) \; \; \Longleftrightarrow \; \; v' \geq p \; \; \; \; \mbox{ and } \; \; \; \;  u(v, \text{not buy}) \geq u(v, \text{buy}) ; \; \Longleftrightarrow \; \; p \geq v \]
Again, these facts imply $v' \geq v$, which is a contraction.

\end{proof}

\eqoneandtwo*


\begin{proof}
Suppose the distribution $F$ of buyers' values is uniform on $[0,1]$, and first fix $q=1$ (i.e. no noise).

At these values, from equation \eqref{eqn:foc}, we know that the price charged in a discriminatory equilibrium, if it exists, is
\begin{align*}
&p^* = \frac{1+ \delta}{2},
\intertext{while the cutoff of types who buy is}
&v^* = \frac{1- \delta}{2}
\end{align*}

Suppose the distribution of period $2$ types is such that $g(v)$ is the step function that is 0 below $v^*$ and 1 above $v^*$.  In this case, the discriminatory equilibrium exists at $q=1$ because we have not added any noise, and consumers who do not buy are certainly type $t_2$ and should be shown ad $B$, while consumers who purchase are almost certainly type $t_1$ and should be shown ad $1$.

The myopic monopoly price under $F$ is $p_M = 1/2$.  Note that if customers were then to purchase myopically, i.e. whenever their value exceeds the price, then the posterior probability of type $t_1$ given that the customer did not purchase is exactly $\delta$. For $\eta < \delta$, we then have that the seller charging the myopic monopoly price, the customers purchasing myopically, and the advertiser showing everyone ad $A$ constitutes a uniform advertising equilibrium A.

In the same game, we can show that both types of equilibria co-exist for a continuous range of $q <1$ as well.  Recall from Proposition \ref{prop.morelikely} that it is sufficient for a uniform advertising equilibrium $A$ to exist if $r(0,p_M,q) > \eta$.  Given myopic behavior on the part of the consumers, and the seller setting the monopoly price in period 1, the advertiser's posterior having seen noisy bit $\bhat=0$ is as follows:
\begin{align*}
r(0,p_M,q) &= \frac{q \int_0^{1/2} \mathbbm{1}_{v>\frac{1-\delta}{2}}dv + (1-q)\int_{1/2}^1 dv}{q F(1/2) + (1-q)(1-F(1/2))} \\
&= \frac{q \frac{\delta}{2} + (1-q)\frac{1}{2}}{q \frac{1}{2} + (1-q)\frac{1}{2}} \\
&= \frac{q \frac{\delta}{2} + (1-q)\frac{1}{2}}{\frac{1}{2}} \\
&=  \delta q + (1-q)
\end{align*}

A uniform advertising equilibrium A will exist whenever $r(0,p_M,q) > \eta$.  That is, for any $q \in [1/2,1]$ satisfying
\[ \delta q + (1-q) > \eta \; \; \; \Longleftrightarrow \; \; \; q > \frac{\eta - 1}{\delta - 1} \]
Since we have already assumed that $\eta < \delta$, then $1 > \frac{\eta -1}{\delta -1}$, so there is a non-empty range of $q$ satisfying this condition.

We now verify that there is a continuous range of $q$ for which there also exists a discriminatory equilibrium of this game.  By equation \eqref{eqn:foc}, the discriminatory equilibrium price $p_1(q)$ must satisfy $p_1(q) - I(p_1(q) + (1-2q)\delta) = 0$, where $I(v) = \frac{1 - F(v)}{f(v)}$.  Plugging in the distribution $F$ as $U[0,1]$,
\begin{align*}
p_1(q) - (1 - p_1(q) - (1-2q)\delta) &= 0 \\
\Longleftrightarrow \; \; \; p_1(q) &= \frac{1}{2} - (1-2q)\frac{\delta}{2} 
\end{align*}
The Period 1 cutoff value is then 
\[ v^*(q) = p_1(q) + (1-2q)\delta = \frac{1}{2} + (1-2q)\frac{\delta}{2} \]

There exists a discriminatory equilibrium at $q \in [1/2,1]$ if both $r(1,v^*(q),q) > \eta$ and $r(0,v^*(q),q) < \eta$.  By construction, $r(1,v^*(q),q) > \eta$ is satisfied for all $q \in [1/2,1]$.
We now compute the advertiser's posterior $r(0,v^*(q),q)$.
\begin{align*}
r(0,v^*(q),q) &= \frac{q \int_0^{v^*(q)} \mathbbm{1}_{v>\frac{1-\delta}{2}}dv + (1-q)\int_{v^*(q)}^1 dv}{q F(v^*(q)) + (1-q)(1-F(v^*(q)))} \\
&= \frac{q (v^*(q) - \frac{1-\delta}{2}) + (1-q)(1-v^*(q))}{q v^*(q) + (1-q)(1-v^*(q))} \\
&= 1 - \frac{q \frac{1-\delta}{2}}{q v^*(q) + (1-q)(1-v^*(q))} \\
&= 1 - \frac{\frac{1}{2}q (1-\delta)}{(1-q) - (1-2q)[\frac{1}{2} + (1-2q)\frac{\delta}{2}]} \\
&= 1 - \frac{\frac{1}{2}q (1-\delta)}{(1-q) - \frac{1}{2}(1-2q) - \frac{1}{2}(1-2q)^2 \delta} \\
&= 1 - \frac{\frac{1}{2}q (1-\delta)}{\frac{1}{2}[1-(1-2q)^2 \delta]} \\
&= 1 - \frac{q (1-\delta)}{1-(1-2q)^2 \delta} \\
\end{align*}

For there to be a discriminatory equilibrium at $q$, it must be the case that $r(0,v^*(q),q) = 1 - \frac{q (1-\delta)}{1-(1-2q)^2 \delta} < \eta$.  Although there is not a nice closed form description of the $q \in [1/2,1]$ satisfying this condition, we note that this expression is differentiable (and thus continuous), everywhere except when 
\[ 1 = (1-2q)^2 \delta \; \; \; \Longleftrightarrow \; \; \; q = \frac{1 \pm \sqrt{\frac{1}{\delta}}}{2} \]
If we restrict $\delta < 1$, then this $q$ will fall outside of our range of interest, and this expression $r(0, v^*(q),q)$ is differentiable on $[1/2,1]$.

We now take the derivative of $r(0, v^*(q),q)$ with respect to $q$, and see that it is negative, so the function is monotone decreasing.
\[ \frac{\partial r(0, v^*(q),q)}{\partial q} = -\frac{(1-\delta)[1-(1-2q)^2 \delta] - [-2\delta (1-2q)(-2)q(1-\delta)]}{[1-(1-2q)^2 \delta]^2} \]
We are only interested in the sign of this expression, and the denominator is clearly positive, so we will proceed only with the numerator (without the negative sign in front of it).
\begin{align*}
(1-\delta)[1-(1-2q)^2 \delta - 4\delta(1-2q)q] &= (1-\delta)[1-\delta (1-2q)(1+2q)] \\
&= (1-\delta)[1-\delta (1-4q^2)] \\
&= (1-\delta)^2 + (1-\delta) 4 \delta q^2
\end{align*}
Since $0 < \delta < 1$, this expression is positive, so plugging it back into the original expression (with a negative in front) means that $\frac{\partial r(0, v^*(q),q)}{\partial q}$ is negative, so $r(0, v^*(q),q)$ is a monotone decreasing function.  This means that there exists a continuous range of $q$ for which the condition $r(0, v^*(q),q) < \eta$ is satisfied.  Since it is satisfied at $q=1$ and $r(0, v^*(q),q)$ is decreasing in $q$, then the condition must be satisfied for all $q \in [1/2,1]$.

To summarize, when $0 < \eta < \delta < 1$, there exists a non-trivial range of $q$, namely $q \in [\frac{\eta - 1}{\delta - 1},1]$, for which there exists both a discriminatory equilibrium and uniform advertising equilibrium A in this game.

\end{proof}

\eqoneandthree*

\begin{proof}
Consider the same example as in Proposition \ref{prop.eq1and2}, where $v \sim U[0,1]$ and $g(v)$ is 1 if $v>\frac{1-\delta}{2}$ and 0 otherwise.  

There will exist a uniform advertising equilibrium B in this game at $q$ if and only if $r(1, p_M, q) < \eta$, where $p_M=\frac{1}{2}$.  The advertiser's posterior in this setting after seeing $\bhat=1$ is
\begin{align*}
r(1,p_M,q) &= \frac{(1-q) \int_0^{1/2} \mathbbm{1}_{v>\frac{1-\delta}{2}}dv + q\int_{1/2}^1 dv}{q F(1/2) + (1-q)(1-F(1/2))} \\
&= \frac{(1-q) \frac{\delta}{2} + q \frac{1}{2}}{q \frac{1}{2} + (1-q)\frac{1}{2}} \\
&= \frac{(1-q) \frac{\delta}{2} + q \frac{1}{2}}{\frac{1}{2}} \\
&= (1-q) \delta + q 
\end{align*}
A uniform advertising equilibrium B exists for all $q$ satisfying
\[ (1-q) \delta + q < \eta \; \; \; \Longleftrightarrow \; \; \; q < \frac{\eta - \delta}{1-\delta} \]
Restricting $\eta > \frac{1+\delta}{2}$ ensures that there is a non-empty interval of $q \in [1/2,1]$ satisfying this condition.

A discriminatory equilibrium exists for any $q$ such that $r(1, v^*(q), q) > \eta$ and $r(0, v^*(q), q) < \eta$.  Recall from the proof of Proposition \ref{prop.eq1and2} that the equilibrium price is $p_1(q) = \frac{1}{2} - (1-2q)\frac{\delta}{2}$ and the equilibrium cutoff value is $v^*(q) =  \frac{1}{2} + (1-2q)\frac{\delta}{2}$.  Also recall that $r(0, v^*(q), q) < \eta$ is equivalent to the condition $1 - \frac{q (1-\delta)}{1-(1-2q)^2 \delta} < \eta$.  We now expand $r(1, v^*(q), q)$.
\begin{align*}
r(1, v^*(q), q) &= \frac{(1-q) \int_0^{v^*(q)} \mathbbm{1}_{v>\frac{1-\delta}{2}}dv + q \int_{v^*(q)}^1 dv}{q F(v^*(q)) + (1-q)(1-F(v^*(q)))} \\
&= \frac{(1-q) (v^*(q) - \frac{1-\delta}{2}) + q(1-v^*(q))}{(1-q) v^*(q) + q(1-v^*(q))} \\
&= 1 - \frac{(1-q) \frac{1-\delta}{2}}{(1-q) v^*(q) + q (1-v^*(q))} \\
&= 1 - \frac{\frac{1}{2}(1-q) (1-\delta)}{q + (1-2q)[\frac{1}{2} + (1-2q)\frac{\delta}{2}]} \\
&= 1 - \frac{\frac{1}{2}(1-q) (1-\delta)}{q + \frac{1}{2}(1-2q) + \frac{1}{2}(1-2q)^2 \delta} \\
&= 1 - \frac{\frac{1}{2}(1-q) (1-\delta)}{\frac{1}{2}[1+(1-2q)^2 \delta]} \\
&= 1 - \frac{(1-q) (1-\delta)}{1+(1-2q)^2 \delta}
\end{align*}

To finish specifying the game parameters, set $\delta = .9$ and $\eta = .984$.  With these parameters, both discriminatory equilibria and uniform advertising equilibria B exist for a continuous range of $q$. We will verify this for one $q$ value in this range, namely $q=.8$.  First, observe that a uniform advertising equilibrium B exists because $r(1, p_M, .8) = .98 < .984 = \eta$.  Next, observe that $r(0, v^*(.8), .8) \approx .88166 < .984 = \eta$ and $r(1, v^*(.8), .8) \approx .98484 > .984 = \eta$.  Thus a discriminatory equilibrium exists as well.
\end{proof}

\cwcaseoneandthree*

\begin{proof}
In any uniform advertising equilibrium B,  the myopic monopoly price is charged, and by assumption, all consumers are shown the inferior ad. Therefore the ex-ante consumer welfare is $\int_{p_M}^1 (v- p_M) f(v) dv$.

At any $q$ where a discriminatory equilibrium exists, the ex-ante consumer welfare is:
\begin{align*}
&\int_{v^*(q)}^1 (v- p_1(q)) f(v) dv + (q (1-F(v^*(q))) + (1- q) F(v^*(q)))\delta\\
=& \int_{v^*(q)}^1 (v- p_1(q)- (1-2q)\delta) f(v) dv + (1-q) \delta.
\end{align*}
Note that the latter term is positive. Further, recall from Proposition \ref{prop.ordering}, $v^*(q) = p_1(q) + (1-2q) \delta$, and that $v^*(q) \leq p_M$. Therefore $\int_{v^*(q)}^1 (v- p_1(q)- (1-2q)\delta) \geq \int_{p_M}^1 (v- p_M) f(v) dv.$

We will now see that if both a discriminatory equilibrium and a uniform advertising equilibrium B exist, then consumers with values $v \in [0, p_M]$ will prefer the discriminatory equilibrium, while the preferences of consumers with values $v \in [p_M, 1]$ depend on the game parameters $\delta$ and $\eta$, as well as the noise level $q$. 

For consumers with values $v \in [0, v^*(q)]$, they prefer the discriminatory equilibrium because they still don't purchase the good in period 1, but they have a chance at the better ad in period 2.  For consumers with values $v \in [v^*(q), p_M]$, under the discriminatory equilibrium, they receive utility $(v - p_1(q)) + \delta q$ for buying the good at price $p_1(q)$ and being shown the better ad with probability $q$.  Under the uniform advertising equilibrium B, they receive utility $0$ for not buying in period 1, and then being shown ad $B$ with probability 1.  For all consumers with value in this range, the discriminatory equilibrium is preferred because 
\[ v - p_1(q) + q \delta > (1-q)\delta > 0, \]
where the first inequality is because the consumer maximized her utility by buying in the discriminatory equilibrium, and the second inequality is because $\delta > 0$ and $q <1$.

For consumers with values $v \in [p_M,1]$, under the discriminatory equilibrium, they receive utility $v - p_1(q) + \delta q$ for buying the good at price $p_1(q)$ and being shown the better ad with probability $q$.  Under the uniform advertising equilibrium B, they receive utility $v-p_M$ for buying at price $p_M$ in period 1, and then being shown ad $B$ with probability 1.  These consumers will prefer the discriminatory equilibrium if and only if 
\[ v - p_1(q) + \delta q > v - p_M \; \; \; \Longleftrightarrow \; \; \; \delta q > p_1(q) - p_M \]
Intuitively, the term $\delta q$ captures the consumer's bonus in period 2 from the possibility of being shown the better ad.  The term $p_1(q) - p_M$ is the additional amount the consumer must pay in period 1 to get the good.  Then the consumer will prefer discrimination whenever the increase in utility from seeing the better ad outweighs the increased price she must pay in period 1.

In the example used in Propositions \ref{prop.eq1and2} and \ref{prop.1and3}, the discriminatory equilibrium is preferred:
\begin{align*} \delta q - (p_1(q) - p_M) &= \delta q - \left(\frac{1}{2} + (2q-1)\frac{\delta}{2} - \frac{1}{2}\right) \\
&= \delta q - (2q-1)\frac{\delta}{2} \\
&= \delta \left(q - \frac{1}{2} \right) > 0
\end{align*}
Since $\delta >0$ and $q > 1/2$, this term is always positive, so discrimination is always preferred in this game.  However, this is not necessarily the case in all games.  We'll next see an example of a game where the players with values $v \in [p_M, 1]$ prefer the uniform advertising equilibrium B.


Consider now the value distribution that is inverse exponential, with parameter $\lambda$, truncated to lie in the domain $[0,1]$.  That is, 
\[ f(v) = \frac{\lambda e^{\lambda v}}{e^{\lambda}-1} \mbox{ and } F(v) = \frac{e^{\lambda v} - 1}{e^{\lambda} - 1} \]
The inverse hazard rate of this function is 
\[ I(v) = \frac{1 - F(v)}{f(v)} = \frac{1 - \frac{e^{\lambda v} - 1}{e^{\lambda} - 1}}{ \frac{\lambda e^{\lambda v}}{e^{\lambda}-1}} = \frac{e^{\lambda} - e^{\lambda v}}{\lambda e^{\lambda v}} = \frac{1}{\lambda} (e^{\lambda(1-v)} - 1) \]

Then the uniform advertising equilibrium $B$ price $p_M$ solves 
\[ p_M - \frac{1}{\lambda} (e^{\lambda(1-p_M)} - 1) = 0 \]
and the discriminatory equilibrium price $p_1(q)$ solves 
\[ p_1(q) - \frac{1}{\lambda} (e^{\lambda(1-v^*(q))} - 1) = 0, \]
where $v^*(q) = p_1(q) + (1-2q) \delta$. 

Choosing $\delta = .5$, $q=1$, and $\lambda=1$, we see that $p_1(q) = .56324$ and $p_M = .557146$.  Thus
\[ \delta q = .055 <  .06094 = p_1(q) - p_M, \]
so these consumers prefer the uniform advertising equilibrium B.


\end{proof}

\cwcaseone*

\begin{proof}
To simplify calculations, suppose again that valuations are distributed on the entire positive real line. Next consider ex-ante equilibrium welfare of the buyer as a function of $q$ in the discriminatory equilibrium:
\begin{align*}
&\int_{v^*(q)}^{\infty} (v-p_1(q)) f(v) dv + q \delta (1-F(v^*(q))) + (1-q) \delta F(v^*(q)).
\intertext{Differentiating with respect to $q$ and collecting terms we have:}
&-f(v^*(q)) v^{*\prime}(q) (v^*(q) - p_1(q) - (1-2q) \delta) + \delta(1-2F(v^*(q))) - p_1'(q) (1-F(v^*(q))).
\intertext{Recall that by the definitions of $v^*(q)$ we have }
&v^*(q) - p_1(q) - (1-2q)\delta =0
\intertext{Therefore, the derivative of ex-ante welfare in a discriminatory equilibrium w.r.t.~$q$ equals}
&\delta(1-2F(v^*(q))) - p_1'(q) (1-F(v^*(q))).
\intertext{Recall from the proof of Lemma \ref{prop.deriv} we have}
&p_1'(q) - I'(p_1 + (1-2q)\delta) (p_1'(q) - 2 \delta) =0,\\
\implies & p_1'(q) = 2q \left(\frac{I'(p_1 + (1-2q)\delta)}{I'(p_1 + (1-2q)\delta)-1} \right).
\intertext{Substituting in, the derivative of welfare w.r.t.~$q$ equals}
&\delta(1-2F(v^*(q))) - 2q \left(\frac{I'(p_1 + (1-2q)\delta)}{I'(p_1 + (1-2q)\delta)-1} \right) (1-F(v^*(q))).
\end{align*}
Note that for $I'$ small (e.g. close to $0$ like an exponential distribution that has constant hazard rate), and $F(v^*)< \frac{1}{2}$, this is positive.


To specify one such game, the example given in Section \ref{s.example} has consumer welfare that is increasing in $q$, for $q$ greater than roughly $0.67$.  (See Figure \ref{fig.csandprofit}.)  Selecting, e.g., $\eta = \frac{1}{2}$ in that game will ensure that discriminatory equilibria exist for all $q$ in that range.

\end{proof}

\end{document}